\documentclass[11pt,a4paper]{article}
\pdfoutput=1

\usepackage[margin=1in]{geometry}
\usepackage{amsmath,amssymb,amsthm,booktabs,color,doi,enumitem,graphicx,latexsym,url,xcolor}
\usepackage[numbers,sort&compress]{natbib}
\usepackage{microtype,hyperref}

\definecolor{citecolor}{HTML}{0000C0}
\definecolor{urlcolor}{HTML}{000080}

\hypersetup{
    colorlinks=true,
    linkcolor=black,
    citecolor=citecolor,
    filecolor=black,
    urlcolor=urlcolor,
}

\newtheorem{theorem}{Theorem}
\newtheorem{lemma}{Lemma}
\newtheorem{corollary}{Corollary}

\theoremstyle{remark}

\newcommand{\namedref}[2]{\hyperref[#2]{#1~\ref*{#2}}}
\newcommand{\sectionref}[1]{\namedref{Section}{#1}}

\newcommand{\theoremref}[1]{\namedref{Theorem}{#1}}
\newcommand{\corollaryref}[1]{\namedref{Corollary}{#1}}
\newcommand{\figureref}[1]{\namedref{Figure}{#1}}
\newcommand{\lemmaref}[1]{\namedref{Lemma}{#1}}

\renewcommand{\vec}[1]{\mathbf{#1}}

\newcommand{\totaln}{\eta}

\DeclareMathOperator{\majority}{majority}
\DeclareMathOperator{\inc}{increment}

\DeclareMathOperator{\polylog}{polylog}
\DeclareMathOperator{\E}{\mathbf E}
\DeclareMathOperator{\N}{\mathbb N}

\newenvironment{mycover}
               {\list{}{\listparindent 0pt
                        \itemindent    \listparindent
                        \leftmargin    0pt
                        \rightmargin   0pt
                        \parsep        0pt}%
                \raggedright
                \item\relax}
               {\endlist}

\begin{document}

\hypersetup{
    pdfauthor={Christoph Lenzen, Joel Rybicki, Jukka Suomela},
    pdftitle={Towards Optimal Synchronous Counting},
}

\begin{mycover}
{\LARGE \textbf{Towards Optimal Synchronous Counting}\par}

\bigskip
\bigskip

\medskip
\textbf{Christoph Lenzen}

\smallskip
{\small Department of Algorithms and Complexity, \\
Max Planck Institute for Informatics\par}

\bigskip
\textbf{Joel Rybicki}

\smallskip
{\small Helsinki Institute for Information Technology HIIT, \\
Department of Computer Science, Aalto University\par

\medskip
Department of Algorithms and Complexity, \\
Max Planck Institute for Informatics\par}

\bigskip
\textbf{Jukka Suomela}

\smallskip
{\small Helsinki Institute for Information Technology HIIT, \\
Department of Computer Science, Aalto University\par}

\end{mycover}

\bigskip
%###
\paragraph{Abstract.}
%###

Consider a complete communication network of $n$ nodes, where the nodes receive a common clock pulse. We study the synchronous $c$-counting problem: given any starting state and up to $f$ faulty nodes with arbitrary behaviour, the task is to eventually have all correct nodes counting modulo $c$ in agreement. Thus, we are considering algorithms that are self-stabilizing despite Byzantine failures. In this work, we give new algorithms for the synchronous counting problem that (1)~are deterministic, (2)~have linear stabilisation time in $f$, (3)~use a small number of states, and (4)~achieve almost-optimal resilience. Prior algorithms either resort to randomisation, use a large number of states, or have poor resilience. In particular, we achieve an \emph{exponential} improvement in the space complexity of deterministic algorithms, while still achieving linear stabilisation time and almost-linear resilience.

\thispagestyle{empty}
\setcounter{page}{0}
\newpage

\section{Introduction}\label{sec:intro}

In this work, we design space-efficient, self-stabilising, Byzantine fault-tolerant algorithms for the \emph{synchronous counting problem}. We are given a complete communication network on $n$ nodes, with arbitrary initial states. There are up to $f$ faulty nodes. The task is to synchronise the nodes so that all non-faulty nodes will count rounds modulo~$c$ in agreement. For example, here is a possible execution for $n = 4$ nodes, $f = 1$ faulty node, and counting modulo $c = 3$; the execution stabilises after $t = 5$ rounds:
\begin{center}
    \begin{tabular}{@{}l@{\ \quad}lllllllllll@{\ \ }ll@{}}
    & \multicolumn{5}{@{}l}{Stabilisation}
    & \multicolumn{7}{l}{Counting} \\
    \cmidrule[\heavyrulewidth](r){2-6}
    \cmidrule[\heavyrulewidth](lr){7-13}
    Node 1: & 2 & 2 & 0 & 2 & 0 & 0 & 1 & 2 & 0 & 1 & 2 & \ldots\! \\
    Node 2: & 0 & 2 & 0 & 1 & 0 & 0 & 1 & 2 & 0 & 1 & 2 & \ldots\! \\
    Node 3: & \multicolumn{11}{@{}l}{\emph{faulty node, arbitrary behaviour}} & \ldots\! \\
    Node 4: & 0 & 0 & 2 & 0 & 2 & 0 & 1 & 2 & 0 & 1 & 2 & \ldots\! \\
    \end{tabular}
\end{center}

Synchronous counting is a coordination primitive that can be used e.g.\ in large integrated circuits to synchronise subsystems so that we can easily implement \emph{mutual exclusion} and \emph{time division multiple access} in a fault-tolerant manner. Note that in this context it is natural to assume that a synchronous clock signal is available, but the clocking system usually will not provide explicit round numbers. Solving synchronous counting thus enables us to construct highly dependable round numbers for subcircuits.

Counting modulo $c = 2$ is closely related to \emph{binary consensus}: given a synchronous counting algorithm one can design a binary consensus algorithm and vice versa \cite{dolev07actions,dolev13counting,dolev15counting}---in particular, many lower bounds on binary consensus apply here as well \cite{dolev85bounds,pease80reaching,fischer82lower}. However, the existing implementations of counting from consensus incur a factor-$f$ overhead in space and message size, rendering them very costly even in small systems.

\begin{table}[b!]
\center
\begin{tabular}{@{}l@{\quad}l@{\quad}l@{\quad}l@{\quad}l@{}}
  \toprule
  resilience & stabilisation time & state bits & deterministic & references \\
  \midrule
  $f<n/3$ & $O(f)$ & $O(f \log f)$ & yes & \cite{dolev07actions} \\ 
  $f<n/3$ & $2^{2(n-f)}$ & 2 & no & \cite{dolev00self-stabilization,dolev04clock-synchronization} \\
  $f<n/3$ & $\min \{ 2^{2f+2} + 1, 2^{O(f^2/n)} \}$ & 1 & no & \cite{dolev15counting} \\
  $f=1$, $n\geq 4$ & 7 & 2 & yes & \cite{dolev15counting} \\
  $f=1$, $n\geq 6$ & 6 & 1 & yes & \cite{dolev15counting} \\
  $f=1$, $n\geq 6$ & 3 & 2 & yes & \cite{dolev15counting} \\
  \midrule
  $f = n^{1-o(1)}$ & $O(f)$ & $O(\log^2 f)$ & yes & {\bf this work} \\
  \bottomrule
\end{tabular}
\caption{Summary of synchronous $2$-counting algorithms. For randomised algorithms, we list the expected stabilisation time. In~\cite{dolev15counting}, further trade-offs between $n$, the stabilization time, and the number of state bits are given (for $f=1$).}\label{table:algorithms}
\end{table}

%###
\paragraph{Prior Work.}
%###

It is fairly easy to design space-efficient \emph{randomised} algorithms for synchronous counting~\cite{dolev00self-stabilization,dolev04clock-synchronization,dolev15counting}: as a simple example, the nodes can just pick random states until a clear majority of them has the same state, after which they start to follow the majority. However, it is much more challenging to come up with space-efficient \emph{deterministic} algorithms for synchronous counting~\cite{dolev07actions,dolev13counting,dolev15counting}, and it remains open to what extent randomisation helps in designing space-efficient algorithms that stabilise quickly. Fast-stabilising algorithms build on a connection between Byzantine consensus and synchronous counting, but require a large number of states per node~\cite{dolev07actions,dolev15counting}.

For small values of the parameters (e.g., $n = 4$ and $f = 1$) the synchronous counting problem is amenable to \emph{algorithm synthesis}: it is possible to use computers to automatically design both space-efficient and time-efficient deterministic algorithms for synchronous counting. For example, there is a computer-designed algorithm that solves the case of $n \ge 4$ and $f = 1$ with only $3$ states per node, and another algorithm that solves the case of $n \ge 6$ and $f = 1$ with only $2$ states per node; both of these are optimal~\cite{dolev13counting,dolev15counting}. Unfortunately, this approach does not scale, as the space of possible algorithms for given parameters $n$ and $f$ grows rapidly with $n$.

In summary, currently no algorithms for synchronous counting are known that simultaneously scale well in terms of resilience $f$, stabilisation time $t$, and memory requirements. Here, it is worth noting that existing solutions with large memory requirements in essence run up to $\Omega(f)$ concurrent instances of consensus, which implies that the respective overhead extends to message size and the amount of local computations as well.

%###
\paragraph{Contributions.}
%###

Our main contribution is a recursive construction that shows how to ``amplify'' the resilience of a synchronous counting algorithm. Given a synchronous counter for some values of $n$ and $f$, we will show how to design synchronous counters for larger values of $n$ and $f$, with a very small increase in time and space complexity. This has two direct applications:
\begin{enumerate}
    \item From a practical perspective, we can use existing computer-designed algorithms (e.g.\ $n = 4$ and $f = 1$) as a building block in order to design efficient deterministic algorithms for a moderate number of nodes (e.g., $n = 36$ and $f = 7$).
    \item From a theoretical perspective, we can now design deterministic algorithms for synchronous counting for any $n$ and for $f \leq n^{1-o(1)}$ faulty nodes, with a stabilisation time of $O(f)$, and with only $O(\log^2 f / \log \log f)$ bits of state per node.
\end{enumerate}
The space complexity is an \emph{exponential} improvement over prior work, and the stabilization time is asymptotically optimal for deterministic algorithms~\cite{fischer82lower}. A summary of the related work and our contributions is given in Table~\ref{table:algorithms}.

In our deterministic algorithms, each node broadcasts its state to all other nodes in each round. However, the small number of state bits bears the promise that the communication load can be reduced further. To substantiate the conjecture that finding algorithms with small state complexity may lead to highly communication-efficient solutions, we proceed to consider a slightly stronger synchronous \emph{pulling model}. In this model, a node may send a request to another node and receive a response in a single round, based on the state of the responding node at the beginning of the round. The cost for the exchange is then attributed to the pulling node; in a circuit, this translates to each node being assigned an energy budget that it uses to ``pay'' for the communication it triggers. In this model, it is straightforward to combine our recursive construction with random sampling to obtain the following results:
\begin{enumerate}
    \item We can achieve the properties of the deterministic algorithm with each node pulling $\polylog n$ messages in each round. The price is that the resulting algorithm retains a probability of $n^{-\polylog n}$ to fail in each round even after stabilisation.
    \item If the failing nodes are chosen independently of the algorithm, we can fix the random choices. This results in a pseudo-random algorithm which stabilises with a probability of $1-n^{-\polylog n}$ and in this case keeps counting correctly.
\end{enumerate}

%###
\paragraph{Structure.}
%###

This paper is organised as follows. In Section~\ref{sec:preliminaries} we formally define the model of computing and the synchronous counting problem. Section~\ref{sec:construction} gives the main technical result---a construction for creating a synchronous counting algorithm of larger resilience from several copies of an algorithm with smaller resilience. Section~\ref{sec:recursion} uses this construction to derive deterministic synchronous counting algorithms with linear stabilisation time and polylogarithmic space complexity. Finally, in Section~\ref{sec:random}, we discuss the pulling model and how randomised sampling can be used to reduce the total number of communicated bits.

\section{Preliminaries}\label{sec:preliminaries}

%###
\paragraph{Model of Computation.}
%###

We consider a fully-connected distributed message-passing system consisting of $n$ processors, also called nodes, with unique identifiers from the set $[n] = \{0, 1, \dots, n-1 \}$. The computation proceeds in synchronous communication rounds, where in each round each processor:
\begin{enumerate}[noitemsep]
 \item \emph{broadcasts} its local state to all processors,
 \item \emph{receives} a vector of messages (that is, states) from all other processors, and
 \item \emph{updates} its local state according to the received messages.
\end{enumerate}
However, the initial state of every node is arbitrary. Moreover, up to $f$ nodes may be \emph{Byzantine}, i.e., exhibit arbitrary behaviour, including to send \emph{different} messages to every node. Thus, different nodes may receive different vectors depending on what the Byzantine nodes do. 

%###
\paragraph{Algorithms.}
%###

A deterministic algorithm in this model is a tuple $\vec A = (X,g,h)$, where $X$ is the set of all possible states for a node, $g \colon [n] \times X^n \to X$ is the state transition function, and $h \colon [n] \times X \to [c]$ maps the internal state of a node to an output value. That is, when node $i \in [n]$ receives a vector $\vec x \in X^n$ of messages, it will update its internal state to $g(i, \vec x) = s$ and output $h(i, s) \in [c]$.

The \emph{space complexity} $S(\vec A)$ of an algorithm $\vec A$ is the total number of bits required to store the state of a node. That is, $S(\vec A) = \left\lceil \log |X| \right\rceil$.

%###
\paragraph{Executions.}
%###

For any given set ${\cal F} \subseteq [n]$ of faulty nodes, we define a projection $\pi_{\cal F}$ as follows: for any received message vector $\vec x \in X^n$, let $\pi_{\cal F}(\vec x)$ be a vector $\vec e$ such that $e_i = *$ if $i \in {\cal F}$ and $e_i = x_i$ otherwise. We call $\pi_{\cal F}(X^n) = \{ \pi_{\cal F}(\vec x) : \vec x \in X^n \}$ the set of \emph{configurations}.  That is, a configuration consists only of the state of all non-faulty nodes.

We say that a configuration $\vec d$ is \emph{reachable} from configuration $\vec e$ if for every non-faulty node $i \notin {\cal F}$ there exists some $\vec x \in X^n$ satisfying $\pi_{\cal F}(\vec x) = \vec d$ and $g(i, \vec x) = d_i$. Intuitively, this means that the Byzantine nodes can send node $i$ such messages that $i$ chooses to switch to state $d_i$ when the system is in configuration $\vec e$.

 An \emph{execution} of a given algorithm $\vec A$ for a given set of faulty nodes, is an infinite sequence of configurations $\xi = \langle \vec e_0, \vec e_1, \vec e_2, \dots \rangle$ such that $\vec e_{t+1}$ is reachable from~$\vec e_t$.

%###
\paragraph{Synchronous Counters.}
%###
We say that an execution $\xi$ of algorithm $\vec A = (X, g, h)$ \emph{stabilises} in time $t$ if there is some $r_0\geq 0$ so that for every non-faulty node $i \in [n]$ the output satisfies
$$
 h(i, x_{t+r,i}) = r - r_0 \bmod c \textrm{ for all } r \ge 0.
$$ 
That is, within $t$ steps all non-faulty nodes agree on a common output and start incrementing their counters modulo $c$ each round. Moreover, we say that $\vec A$ is a \emph{synchronous $c$-counter} with \emph{resilience} $f$ if there exists a $t$ such that for every ${\cal F} \subseteq [n]$, $|{\cal F}| \le f$ all executions of $\vec A$ stabilise in $t$ rounds. We say that the \emph{stabilisation time} of algorithm $\vec A$ is $T(\vec A) \le t$. In the following, we use $\mathcal{A}(n,f,c)$ to denote the family of synchronous $c$-counters that run on $n$ nodes. 

\section{Boosting Resilience}\label{sec:construction}

In this section, we show that given a family of synchronous $c$-counters for a small number of nodes $n$ and resilience $f$, we can construct a new family of $C$-counters of a larger resilience without increasing the number of nodes in the network or the stabilisation time by too much. More precisely, we prove the following main result.
\pagebreak % LAYOUT
\begin{theorem}\label{thm:layering}
Given $n,f\in \N$, pick new parameters $N,F,C\in \N$, where
\begin{itemize}[noitemsep]
  \item the number of nodes is $N=kn$ for some number of blocks $3\leq k\in \N$,
  \item the resilience is $F< (f+1)m$, where we abbreviate $m=\lceil k/2\rceil$, and
  \item the counter size is $C>1$.
\end{itemize}
Choose any $c\in \N$ that is a multiple of $3(F+2) (2m)^k$. Then for any $\vec{A}\in \mathcal{A}(n, f, c)$ there exists a $\vec B \in \mathcal{A}(N,F,C)$ with
\begin{align*}
 T(\vec B) &\leq  T(\vec A) + 3(F+2) (2m)^k, \\
 S(\vec B) &= S(\vec A) + \left\lceil \log (C+1) \right\rceil + 1.
\end{align*}
\end{theorem}

%###
\subsection{High-level Idea of the Construction}
%###
Given a suitable counter $\vec A \in \mathcal{A}(n,f,c)$, we construct a larger network with $N=kn$ nodes and divide the nodes in $k$ blocks of $n$ nodes. Each block runs a copy of $\vec A$ that is used to (1)~determine a ``leader block'' and (2)~output a consistent round counter \emph{within} a block.

Due to the bound on $F$, only a minority of the blocks (fewer than $m$) contains more than $f$ faulty nodes. Thus, the counters of a majority of the blocks stabilise within $T(\vec A)$ rounds. By dividing the counters within each block, we let them ``point'' to one of $m$ possible leader blocks for a fairly large number of consecutive rounds. Block $i$ switches through leaders by a factor of $2m$ faster than block $i+1$. This ensures that, eventually, all stabilised counters point to the same leading block for $3(F+2)$ rounds.

Using a majority vote on the pointers to leader blocks, we can make sure that all correct nodes---also those in blocks with more than $f$ faulty nodes---will recognise the same block as the leader for sufficiently long. In fact, eventually this will happen for each of the $m$ blocks that may become leaders, one of which must have fewer than $f$ faulty nodes. In particular, its counter will be stabilised and can be ``read'' by all other nodes based on taking the majority value within that block. Hence, we can guarantee that at some point all nodes agree on a common counter value for at least $3(F+2)$ rounds. This value is used to control an execution of the well-known \emph{phase king protocol}~\cite{berman89consensus}, which solves consensus in $\Theta(F)$ rounds in systems of $N>3F$ nodes. As the constraint $F<(f+1)m$ also ensures that $F<N/3$, we can use the phase king protocol to let all nodes agree on the current value of the $C$-counter that is to be computed. It is straightforward to guarantee that, once this agreement is achieved, it will not be lost again and the counter is incremented by one modulo $C$ in each round.

%###
\subsection{Setup}
%###

Our goal is to construct a synchronous $C$-counter that runs in a network of $N=kn$ nodes and tolerates $F < (f+1)m$ faults. We divide the nodes into $k$ blocks, each of size $n$. We identify each node $v \in [kn]$ with a tuple $(i,j) \in [k] \times [n]$. Thus, node $v = (i,j)$ is the $j$th node of block~$i$. Blocks that contain more than $f$ faulty nodes are said to be \emph{faulty}. Otherwise, a block is \emph{non-faulty}.

We will have each block $i$ run a synchronous counter $\vec A_i$ constructed as follows. Define $\tau=3(F+2)$ and $c_i = \tau (2m)^{i+1}$ for $i \in [k]$. Let $\vec A = (X,g,h) \in \mathcal{A}(n,f,c)$ be the given synchronous $c$-counter, where $c = \alpha \tau (2m)^k$ for some integer $\alpha$. Now for any $i \in [k]$, we obtain a synchronous $c_i$-counter $\vec A_i = (X, g, h_i)$ by defining the output function as $h_i(x) = h(x) \bmod c_i$. That is, $\vec A_i$ outputs the output of $\vec A$ modulo $c_i$. It follows that $T(\vec A_i) = T(\vec A)$ and $S(\vec A_i) = S(\vec A)$. Thus, if block $i$ is non-faulty, then $\vec A_i$ will stabilise in $T(\vec A)$ rounds.

We interpret the value of the counter of (non-faulty, stabilised) block $i$ as a tuple $(r, y) \in [\tau] \times [(2m)^{i+1}]$, where $r$ is incremented by one modulo $\tau$ each round, and $y$ is incremented by one whenever $r$ ``overflows'' to $0$. We refer to the counter value that node $(i,j)$ currently has according to $\vec A_i$ as $(r[i,j],y[i,j])$. Note that this value can be directly inferred from its state, so by broadcasting its state $(i,j)$ implicitly announces its current counter value to all other nodes. We stress that there is no guarantee whatsoever on how these variables behave in faulty blocks, even for non-faulty nodes in faulty blocks---the only guarantee is that non-faulty nodes in non-faulty blocks will count correctly after round $T(\vec A)$. 

We define the short-hand
\[
b[i,j] = \left \lfloor \frac{y[i,j]}{(2m)^{i}} \right \rfloor \bmod m.
\]
The value $b[i,j]$ indicates the block that the nodes in block $i$ currently consider to be the ``leader block'' by interpreting the counter given by $\vec A_i$ appropriately. 

Let $b_q[i,j]$ and $r_q[i,j]$ denote the values of $b[i,j]$ and $r[i,j]$ in round $q$. We will now show that after stabilisation, a non-faulty block $i$ will within $c_i$ rounds point to every block $\beta \in [m]$ for at least $\tau$ consecutive rounds. For notational convenience, we let $c_{-1} = \tau$.

\begin{lemma}\label{lemma:consecutive}
Let $i \in [k]$ be a non-faulty block and $t \ge T(\vec A)$. For any $\beta \in [m]$, there exists some $t \le w \le t + c_i - c_{i-1}$ such that if $(i,j)$ is non-faulty, then $b_q[i,j] = \beta$ for $w \le q < w + c_{i-1}$.
\end{lemma}
\begin{proof}
By round $t$, the counter $\vec A_i$ has stabilised and all non-faulty nodes $(i,j)$ agree. First observe that since $\vec A_i$ is a $c_i$-counter, $b[i,j]$ cycles through the set $[m]$ twice in $c_i$ rounds. Moreover, once $b[i,j]$ changes its value, it will keep the new value for $c_{i-1} = \tau(2m)^{i}$ consecutive rounds.

More formally, let $(i,j)$ be a non-faulty node and $u \ge t$ be the minimal $u$ such that $r_u[i,j] = 0$ and $y_u[i,j] \in \{0, c_i/(2 \tau)\}$. In particular, at round $u$ we have $b_u[i,j] = 0$. Moreover, $u \le t + c_i/2$. Since $b[i,j]$ retains the same value for $c_{i-1}$ consecutive rounds, we get that at round $w = u + (\beta-1)c_{i-1}$ the value $b[i,j]$ changes to $\beta$. Now we have that $b_{w'}[i,j] = \beta$ for $w \le w' < w + c_{i-1}$. A simple check confirms that $w + c_{i-1} \le t + c_i$.
\end{proof}

Using this lemma, we can show that after stabilisation, all non-faulty nodes in non-faulty blocks will within $c_k$ rounds point to each $\beta \in [m]$ simultaneously for at least $\tau$ rounds.

\begin{lemma}\label{lemma:rounds}
Let $t \ge T(\vec A)$ and $\beta \in [m]$. There exists some $t \le u \le t+c_k - \tau$ such that if $(i,j)$ is a non-faulty node in a non-faulty block, then $b_q[i,j] = \beta$ for $u \le q < u + \tau$.
\end{lemma}
\begin{proof}
We prove the statement for the special case that all blocks are non-faulty. As there is no interaction between the algorithm $\vec A_i$ of different blocks, the the lemma then follows by simply excluding all nodes in faulty blocks. We prove the following claim using induction; see Figure~\ref{fig:lemmaproof} for illustration.

Claim: For any round $t \ge T(\vec A)$, block $\beta \in [m]$ and $h \in [k]$, there exists a round $t \le u(h) \le t + c_h - \tau$ such that for all non-faulty nodes $(i,j)$ where $i \le h$ we have $b_q[i,j] = \beta$ for $u(h) \le q < u(h)+\tau$.

In particular, the lemma follows from the case $h=k$ of the above claim. For the base case of the induction, observe that case $h=0$ follows from \lemmaref{lemma:consecutive} as $c_{-1} = \tau$. For the inductive step, suppose the claim holds for some $h \in [k-1]$ and consider the non-faulty block $h+1$. By \lemmaref{lemma:consecutive} there exists $t \le w \le t + c_{h+1} - c_h$ such that $b_q[h+1,j] = \beta$ for all $w \le q < w + c_h$ where $(h+1,j)$ is a non-faulty node. Applying the induction hypothesis to $w$ we get $u(h)$. Setting $u(h+1) = u(h)$ yields that $u(h+1) \le w + c_h - \tau \le t + c_{h+1} - \tau$. This proves the claim and the lemma follows. 
\end{proof}

\begin{figure}
\begin{center}
\includegraphics[page=2]{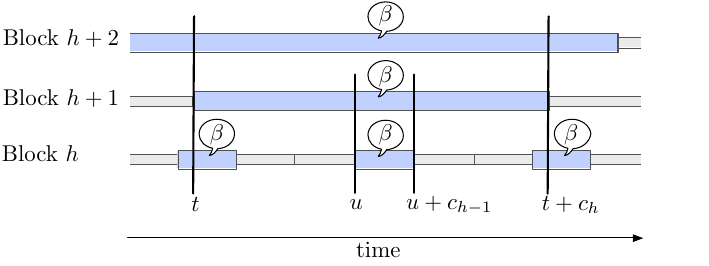}
\end{center}
\caption{Counters in non-faulty blocks will eventually coincide. The picture illustrates the output of $b[\cdot]$ for three blocks $i \in \{h,h+1,h+2\}$ running $\tau(2m)^{i+1}$-counters with base $(2m)=6$. For every $\beta \in [m]$, we can find a interval (blue segments) where all non-faulty blocks point to the same value $\beta$ for sufficiently long, even though the clocks may cycle at different time points.
\label{fig:lemmaproof}}
\end{figure}

%###
\subsection{Voting Blocks}
%###

Next, we define the voting scheme for all the blocks. Define the $\majority$ operation for every message vector~$\vec x \in X^{kn}$ as follows:
\[
\majority \vec x = \begin{cases}
               a &\textrm{if } a \textrm{ is contained in } \vec x \textrm{ more than } kn/2 \textrm{ times}, \\
               * &\textrm{otherwise},
\end{cases}
\]
where the symbol $*$ indicates that the function may evaluate to an arbitrary value, including different values at different non-faulty nodes. We use the following short-hands as local variables:
\begin{align*}
 b^{i} &= \majority\{ b[i,j] \mid j \in [n] \}, \\
 B    &= \majority\{ b^{i} \mid i \in [k] \}, \\
 R    &= \majority\{ r[B,j] \mid j \in [n] \}.
\end{align*}
Note that these functions can be locally computed from the received state vectors by checking for a majority and defaulting to, e.g., $0$, when no such majority is found: by definition, the $\majority$ function may return an arbitrary value if $kn/2$ or fewer correct nodes ``vote'' for the same value; as non-faulty nodes broadcast the same state to all nodes, there can only be one such majority value.

In words, $b^{i}$ denotes the block which the nodes in block $i$ support as a leader; different correct nodes may ``observe'' different values of $b^i$ only if $i$ is a faulty block or $\vec A_i$ has not yet stabilised. As $F<(f+1)m=(f+1)\lceil k/2\rceil$, a majority of the blocks is non-faulty. Hence, if all non-faulty blocks support the same leader block $\beta \in [m]$, then $B$ evaluates to $\beta$ at all correct nodes. By \lemmaref{lemma:rounds}, this is bound to happen eventually. Finally, $R$ denotes the round counter of block $B$, which is ``read'' correctly by all non-faulty nodes if $B$ is non-faulty.

Analogously to before, let $i'\in [k]$ and denote by $B_q[i,j]$ and $R_q[i,j]$, respectively, the values to which the above functions evaluate in round $q$ at node $(i,j)$. Then we can conclude from \lemmaref{lemma:rounds} that eventually all non-faulty nodes agree on $R$ for $\tau$ rounds.
\begin{lemma}\label{lemma:round-counter}
There is a round $t \le T(\vec A) + c_k - \tau$ such that:
\begin{enumerate}[label={(\alph*)},noitemsep]
  \item $R_q[i,j] = R_q[i',j']$ for any $t \le q < t + \tau$ and non-faulty nodes $(i,j)$ and $(i',j')$.
  \item $R_{q+1}[i,j]=R_{q}[i,j] + 1 \bmod \tau$ for any $t \le q< t+\tau - 1$ and non-faulty node $(i,j)$.
\end{enumerate}
\end{lemma}
\begin{proof}
As $F<(f+1)m$, there is a non-faulty block $\beta \in [m]$. By applying \lemmaref{lemma:rounds} to round $T(\vec A)$, there is a round $t \le T(\vec A) + c_k - \tau$ such that $b_q[i,j] = \beta$ for each $t \le q < t + \tau$ and non-faulty node $(i,j)$ in a non-faulty block $i$. Therefore, $b_q^{i} = \beta$ for all non-faulty blocks $i$. As $F<(f+1)m$, the number of faulty blocks is at most $m-1=\lceil k/2\rceil-1<k/2$, and thus, the majority vote yields $B=\beta$. 

Since $T(\vec A) \le t \le q < t+\tau$, we have that block $\beta$ has stabilised by round $t$ and therefore, $r_q[\beta,j]=r_q[\beta,j']$ for all non-faulty nodes $(\beta,j)$, $(\beta,j')$ in the non-faulty block $\beta$. Moreover, as block $\beta$ is non-faulty, it contains at most $f$ faulty nodes. In particular, it must be that $f < n/3 < n/2$ as otherwise counting cannot be solved and 
$\vec A \in \mathcal{A}(n,f,c)$. Hence, a majority vote yields $R=r_q[\beta,j]$, where $(\beta,j)$ is any non-faulty node in block $\beta$ proving claim~(a).

To show (b), observe that since block $\beta$ is non-faulty and has stabilised by round $t$ we have that $R_q[i,j]=r_q[\beta,j']$. Moreover, non-faulty nodes in block $\beta$ increment $r[\beta,j]$ by one modulo $\tau$ in the considered interval.
\end{proof}

%###
\subsection{Executing the Phase King}
%###

We have now built a voting scheme that allows the nodes to eventually agree on common counter for $\tau$ rounds. Roughly speaking, what remains is to use this common counter to control a non-self-stabilizing $F$-resilient $C$-counting algorithm. 

We require that this algorithm guarantees two properties. First, all non-faulty nodes reach agreement and start counting correctly within $\tau$ rounds provided that the underlying round counter is consistent. Second, if all non-faulty nodes agree on the output, then the agreement persists \emph{regardless} of the round counter's value. It turns out that a straightforward adaptation of the classic phase king protocol~\cite{berman89consensus} does the job.

From now on, we refer to nodes by their indices $v \in [N]$. The phase king protocol (like any consensus protocol) requires that $F<N/3$. It is easy to verify that this follows from the preconditions of \theoremref{thm:layering}.

Denote by $a[v] \in [C] \cup \{ \infty \}$ the output register of the algorithm, where $\infty$ is used as a ``reset state''. There is also an auxiliary register $d[v] \in \{0,1\}$. Define the following short-hand for the increment operation modulo $C$:
\[
\inc a[v] = \begin{cases}
a[v] \gets a[v] + 1 \bmod C & \text{if }a[v]\neq \infty,\\
\mbox{no action} & \text{if }a[v]=\infty.
\end{cases}
\]

\begin{table}
\centering
\begin{tabular}{@{}l@{\qquad}r@{\ \ }l@{}}
\toprule
Set & \multicolumn{2}{@{}l@{}}{Instructions} \\
\midrule
$I_{3\ell}$:
& 1. & If fewer than $N-F$ nodes sent $a[v]$, set $a[v] \gets \infty$. \\
& 2. & $\inc a[v]$. \\
\midrule
$I_{3\ell+1}$:
& 1. & Let $z_j = |\{u \in [N] : a[u] = j\}|$ be the number of $j$ values received. \\
& 2. & If $z_{a[v]} \ge N-F$, set $d[v] \gets 1$. Otherwise, set $d[v] \gets 0$. \\
& 3. & Set $a[v] \gets \min \{ j : z_j > F\}$. \\
& 4. & $\inc a[v]$. \\
\midrule
$I_{3\ell+2}$:
& 1. & If $a[v] = \infty$ or $d[v] = 0$, then set $a[v] \gets \min\{C, a[\ell]\}$. \\
& 2. & Set $d[v] \gets 1$ and $\inc a[v]$. \\
\bottomrule
\end{tabular}
\caption{The instruction sets for node $v\in [N]$ in the phase king.}\label{tab:instr}
\end{table}

For $\ell \in [F+2]$, we define the instruction sets listed in Table~\ref{tab:instr}. First, we show that if these instructions are executed in the right order by all non-faulty nodes for a non-faulty leader $\ell \in [F+2]$, then agreement on a counter value is established.
\begin{lemma}\label{lemma:establish_agreement}
Suppose that for some non-faulty node $\ell \in [F+2]$ and a round $q$, all non-faulty nodes execute instruction sets $I_{3\ell}$, $I_{3\ell+1}$, and $I_{3\ell+2}$ in rounds $q-2$, $q-1$, and $q$, respectively. Then $a_{q+1}[v]=a_{q+1}[u]\neq \infty$ for any two non-faulty nodes $u,v\in [N]$. Moreover, $d_{q+1}[v]=1$ at each non-faulty node.
\end{lemma}

\begin{proof}
This is essentially the correctness proof for the phase king algorithm. Without loss of generality, we can assume that the number of faulty nodes is exactly $F$. By assumption, we have $F<N/3$ and hence $2(N-2F)>N-F$. It follows that it is not possible that two non-faulty nodes $v,u\in [N]$ satisfy both $a_{q-1}[v],a_{q-1}[u]\in [C]$ and $a_{q-1}[v]\neq a_{q-1}[u]$: this would imply that there are at least $N-2F$ non-faulty nodes $w$ that had $a_{q-2}[w] + 1 \bmod C =a_{q-1}[v]$ and the same number of non-faulty nodes $w'$ with $a_{q-2}[w'] + 1 \bmod C =a_{q-1}[u]$; however, there are only $N-F<2(N-2F)$ non-faulty nodes. Therefore, there is some $x\in [C]$ so that $a_{q-1}[v]\in \{x,\infty\}$ for all non-faulty nodes $v$. Checking $I_{3\ell+2}$ and exploiting that $2(N-F)>N-F$ once more, we see that this implies that also $a_q[v]\in \{x + 1 \bmod C,\infty\}$ for some  any non-faulty node~$v$.

We need to consider two cases. In the first case, all non-faulty nodes execute the first instruction of $I_{3\ell+2}$ in round $q$. Then $a_{q+1}[v]=\min\{C, a_q[\ell]\}+1\bmod C$ for any non-faulty node~$v$. In the second case, there is some node $v$ not executing the first instruction of $I_{3\ell+2}$. Hence, $d_q[v]=1$, implying that $v$ computed $z_{a_{q-1}[v]}\geq N-F$ in round $q-1$. Consequently, at least $N-2F>F$ non-faulty nodes $u$ satisfy $a_{q-1}[u]=a_{q-1}[v]$. We infer that $a_q[v] = x'$ for all non-faulty nodes $v$: the third instruction of $I_{3\ell+1}$ must evaluate to $x' \in [C]$ at all non-faulty nodes. Clearly, this implies that $a_{q+1}[v]=a_{q+1}[u]\neq \infty$ for non-faulty nodes $v,u$, regardless of whether they execute the first instruction of $I_{3\ell+2}$ or not. Trivially, $d_{q+1}[v]=1$ at each non-faulty node $v$ due to the second instruction of $I_{3\ell+2}$.
\end{proof}

Next, we argue that once agreement is established, it persists---it does not matter any more which instruction sets are executed.
\begin{lemma}\label{lemma:maintain_agreement}
Assume that $a_q[v] = x \in [C]$ and $d_q[v] = 1$ for all non-faulty nodes $v$ in some round~$q$. Then $a_{q+1}[v] = x+1 \bmod c$ and $d_{q+1}[v]=1$ for all non-faulty nodes $v$.
\end{lemma}
\begin{proof}
Each node will observe at least $N-F$ nodes with counter value $x$, and hence at most $F$ nodes with some value $y\neq x$. For non-faulty node $v$, consider all possible instruction sets it may execute. 

First, consider the case where instruction set $I_{3\ell}$ is executed. In this case, $v$ increments $x$, resulting in $a_{q+1}[v]=x+1\bmod C$ and $d_{q+1}[v] = 1$. Second, executing $I_{3\ell+1}$, node $v$ evaluates $z_x \ge N-F$ and $z_y \le F$ for all $y \neq x$. Hence it sets $d_{q+1}[v]=1$ and $a_{q+1}[v]=x + 1 \bmod C$. Finally, when executing $I_{3\ell+2}$, node $v$ skips the first instruction and sets $d_{q+1}[v]=1$ and $a_{q+1}[v]=x+1 \bmod C$.\qedhere
\end{proof}

%###
\subsection{Proof of Theorem~\ref{thm:layering}}
%###

We can now prove the main result. As shown in \lemmaref{lemma:round-counter} we have constructed a $\tau$-counter that will remain consistent at least $\tau$ rounds. This is a sufficiently long time for the nodes to execute the phase king protocol in synchrony. This protocol will stabilise the $C$-counter for the network of $N$ nodes. More precisely, each node $(i,j)$ runs the following algorithm:
\begin{enumerate}[noitemsep]
 \item Update the state of algorithm $\vec A_i$.
 \item Compute the counter value $R$.
 \item Update state according to instruction set $I_{R}$ of the phase king protocol.
\end{enumerate}

By \lemmaref{lemma:round-counter}, there is a round $t\leq T(\vec A) + c_k - \tau$ so that the variables $R_q[i,j]$ meet the requirements of $\tau$-counting for rounds $t \le q < t + \tau$. For each round $q$, all non-faulty nodes execute the same set of instructions. In particular, as $\tau = 3(F+2)$, no matter from which value the $\tau$-counting starts, for at least $F+1$ values $\ell \in [F+2]$  the instruction sets $I_{3\ell}$, $I_{3\ell+1}$, and $I_{3\ell+2}$, in this order, will be jointly executed by all non-faulty nodes at some point during rounds $t, \ldots, t + \tau - 1$. 

As there are only $F$ faulty nodes, there are at least two non-faulty nodes $\ell \in [F+2]$. Thus, the prerequisites of \lemmaref{lemma:establish_agreement} are satisfied in some round $q \le t + \tau \leq T(\vec A)+c_k$. By an inductive application of \lemmaref{lemma:maintain_agreement}, we conclude that the variables $a_q[v]$ are valid outputs for $C$-counting, and therefore, we have indeed constructed an algorithm $\vec B\in {\cal A}(N,F,C)$.

The bound on $q$ yields that $T(\vec B) \le T(\vec A)+c_k=T(\vec A) + 3(F+2)(2m)^k$. Concerning the state complexity, observe that each non-faulty node $(i,j)$ needs the memory for executing (1)~the algorithm $\vec A_i$, which needs $S(\vec A_i)=S(\vec A)$ bits of memory, and (2)~the phase king protocol, which needs $\lceil \log(C+1)\rceil$ bits to store $a_q[v] \in [C] \cup \{\infty\}$ and one additional bit to store $d_q[v]$.\qed

\section{The Recursive Construction}\label{sec:recursion}

In this section, we show how to use \theoremref{thm:layering} recursively to construct synchronous $c$-counters with a near-optimal resilience, linear stabilisation time, and a small number of states (see \figureref{fig:layers} for an illustration). First, we show how to satisfy the preconditions of \theoremref{thm:layering} in order to start the recursion. Then we demonstrate the principle by choosing a fixed value of $k$ throughout the construction; this achieves a resilience of $\Omega(n^{1-\varepsilon})$ for any constant $\varepsilon > 0$. However, as the number of nodes in the initial applications of \theoremref{thm:layering} is small, better results are possible by starting out with large values of $k$ and decreasing them later. This yields an algorithm with a resilience of $n^{1-o(1)}$ and $O(f)$ stabilisation time using $O(\log^2 f / \log \log f + \log c)$ state bits.

%###
\subsection{The Base Case}
%###
To apply \theoremref{thm:layering}, we need counters of resilience $f>0$.  For example, one can use the space-efficient $1$-resilient counters from \cite{dolev15counting} as base of the construction.  Alternatively, we can use as a starting point trivial counters for $n = 1$ and $f = 0$. Then we can apply the same construction as in \theoremref{thm:layering} with the parameters $n=1$, $f=0$, $k=N$, and $F<N/3$. The same proof goes through in this case and yields the following corollary. Note that here the resilience is optimal but the algorithm is inefficient with respect to the stabilisation time and space complexity.
\begin{corollary}\label{cor:base}
 For any $c > 1$, there exists a synchronous $c$-counter with optimal resilience $f<n/3$ that stabilises in $f^{O(f)}$ rounds and uses $O(f \log f + \log c)$ bits of state.
\end{corollary}
\begin{proof}
 For any $f > 0$, we can construct a $f$-resilient counter for $3f+1$ nodes. We use the trivial 0-resilient counter for one node as the base case for Theorem~\ref{thm:layering} and set $k=3f+1$, that is, each block consists of a single node. Theorem~\ref{thm:layering} yields an $f$-resilient algorithm that stabilises in $(3f+2)(2m)^k = f^{O(f)}$ rounds, where the $O(f)$ term is $(3+o(1))f$.
\end{proof}

%###
\subsection{Using a Fixed Number of Blocks}
%###

For the sake of simplicity, we will first discuss the recursive construction for a fixed value of $k$ here. Improved resilience can be achieved by varying $k$ depending on the level of recursion which we show afterwards.

\begin{figure}
\includegraphics[page=1,width=\columnwidth]{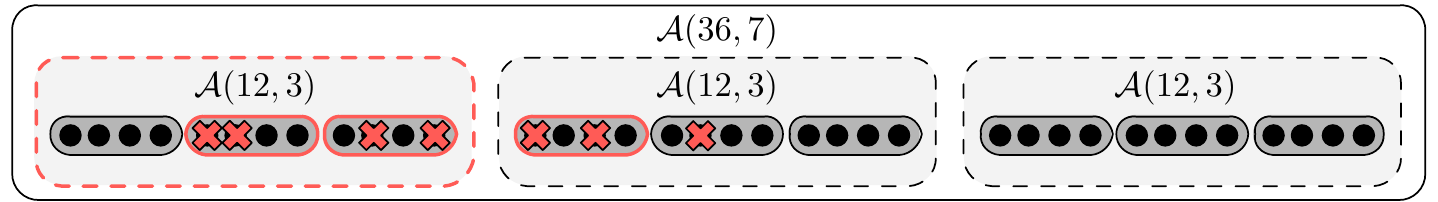}
\caption{Recursive application of our construction using $k=3$ blocks. The filled circles represent the nodes. The groups of four nodes run $1$-resilient counters. On top of this, the dashed groups run $3$-resilient counters. At the top-most layer, the nodes run a $7$-resilient counter. Faulty blocks are drawn with a red border, whereas faulty nodes are marked with a red cross.
\label{fig:layers}}
\end{figure}

\begin{theorem}\label{thm:fixed-blocks}
 Let $1 > \varepsilon > 0$ and $2\leq c \in \N$. There exists a synchronous $c$-counting algorithm with a resilience of $f = \Omega(n^{1-\varepsilon})$ that stabilises in $O(2^{2^{1/\varepsilon}} f)$ rounds and uses $O(2^{1/\varepsilon}\log f + \log^2 f + \log c)$ bits of state per node.
\end{theorem}
\begin{proof}
Fix $0 < \varepsilon < 1$ and let $k = 2h \ge 4$ be minimal such that $\varepsilon \ge 1/\log h$. Assume $f = 2^j$ for some $j \ge h$ and let $L = \log f / \log h$; w.l.o.g., assume that $L\in \N$. We will analyse how many nodes are required to get to a resilience of $\Omega(n^{1-\varepsilon})$ by applying \theoremref{thm:layering} for $L \le \varepsilon \log f$ iterations.

For all $i \ge 0$, let $f_{i} = h^i$ and $n_i = 4k^i$. At iteration $i+1$ we use \theoremref{thm:layering} to construct algorithms in $\mathcal{A}(n_{i+1}, f_{i+1}, c)$ (for any $c$) using $n_i$ and $f_i$ as the input parameters. Since $f_{i+1} = f_i h < (f_i+1) \lceil k/2 \rceil$ and $n_{i+1} = kn_i$, the conditions of \theoremref{thm:layering} are satisfied. To start the recursion, we will use an algorithm with parameters $f_0 = 1$ and $n_0 = 4$. By Corollary~\ref{cor:base}, such an algorithm with a stabilisation time of $T_0 = O(1)$ and a space complexity of $S_0 = O(\log c)$ exists.

Every iteration increases the resilience by a factor of at least $h=k/2$. After $L$ iterations, we tolerate at least $f = f_L = h^L$ failures using $n = n_L = 4k^L$ nodes. This gives 
\[
\frac{n}{f} = 4 \left(\frac{k}{k/2}\right)^L = 4 \cdot 2^L \le 8 f^{\varepsilon} <  8 n^{\varepsilon}.
\] 
and it follows that the resilience is $\Omega(n^{1 - \varepsilon})$.

It remains to analyse the stabilisation time and space complexity of the resulting algorithm; both follow from Theorem~\ref{thm:layering}. The stabilisation time of layer $i+1$ is 
\[
 T_{i+1} \le T_i + 3(f_{i+1} + 2)h^k,
\]
where $h = \lceil k/2 \rceil$. From the definition of $f_i$, we get the bound $\sum_{i=0}^{L} f_i = O(f_L) = O(f)$. Therefore, the overall stabilisation time is 
\[
   T = T_L \le O\!\left( h^k \sum_{i=0}^{L} f_i \right)  = O( h^k f).
\]
From \theoremref{thm:layering}, we get that the space complexity of layer $i+1$ is 
\[
 S_{i+1} \le S_i + \lceil \log(3(f_{i+2}+2)h^k + 1) \rceil + 1 = S_i + O(k \log h + \log f).
\]
As $L\leq \varepsilon \log f$, the total number of memory bits is then bounded by 
\[
 \sum_{i=0}^{L} O(k \log h + \log f) + O(\log c)\leq O(\varepsilon k \log h \log f + \log^2 f + \log c).
\]
Recall that $h=k/2$ is minimal such that $\varepsilon\geq 1/\log h$. Thus $k = O(2^{1/\varepsilon})$ and $\log h = O(1/\varepsilon)$, yielding the claimed bounds on time and space complexity. 
\end{proof}
Choosing a constant $\varepsilon$, we arrive at the following corollary.
\begin{corollary}
 For any constant $1 > \varepsilon > 0$ and any $2\leq c \in \N$, there exists a synchronous $c$-counter with resilience $f = \Omega(n^{1-\varepsilon})$ that stabilises in $O(f)$ rounds and uses $O(\log^2 n + \log c)$ bits of state.
\end{corollary}

%###
\subsection{Varying the Number of Blocks}
%###

Obviously, the factor $2^{2^{1/\varepsilon}}$ makes the previous construction impractical unless $1/\varepsilon$ is small. However, it turns out that we can still achieve good resilience without a doubly-exponential blow-up in the stabilisation time by carefully varying the number of blocks at each level. 

\begin{theorem}\label{thm:adaptive}
 For any $c>1$, there exist synchronous $c$-counters with a resilience of $f=n^{1-o(1)}$ that stabilises in $O(f)$ rounds and uses $O(\log^2 f  / \log \log f + \log c)$ bits of space per node.
\end{theorem}
Proving this theorem boils down to choosing $k$ in each iteration as large as possible without violating the bound on the stabilization time. We again rely on Theorem~\ref{thm:layering}, but instead of using a fixed number of blocks at each iteration, we divide the construction into $P$ \emph{phases}. During each phase, we use a different number of blocks and iterations of \theoremref{thm:layering}. The goal is to have the running time of the last phase dominate the running time of earlier phases. 

We set each phase $1 \le p \le P$ to use $k_p = 4\cdot2^{P-p} > 3$ blocks per layer and then iterate \theoremref{thm:layering} exactly $R_p = 2k_p$ times. During phase $p$, we use iteration $1 \le i+1 \le R_p$ to get an algorithm from \theoremref{thm:layering} that tolerates 
\[
f_{i+1,p} = \frac{f_{i,p}k_p}{2} < (f_{i,p} + 1) \left \lceil \frac{k_p}{2}\right\rceil 
\]
failures, where $f_{0,p} = f_{R_{p-1}, p-1}$ and $f_{R_0,0} = f_{0,0}=1$. Thus, for any $p$ and $i$, the values $k_p$ and $f_{i,p}$ satisfy the conditions of Theorem~\ref{thm:layering}. Again, to start the recursion we may use any algorithm tolerating a single fault among $4$ nodes giving $f_{0,0}=1$.

Now, every phase $p$ increases the resilience by a factor of
\[
 d_p = \left(\frac{f_{i+1,p}}{f_{i,p}}\right)^{R_p} = \left(\frac{k_p}{2}\right)^{R_p}.
\]
As there are total of $P$ phases, this means that the total resilience and number of nodes $n$ are given by
\[
 f = f_P = \prod_{p=1}^P d_p = \prod_{p=1}^P \left(\frac{k_p}{2}\right)^{R_p}
\quad \text{and} \quad n = n_P = 4 \prod_{p=1}^P k_p^{R_p}.
\]
In order to get resilience of $f \ge n^{1-\varepsilon}$, where $\varepsilon = o(1)$, we want to ensure that
\[
 \frac{n}{f} = 4 \prod_{p=1}^P 2^{r_p} = 4 \prod_{j=0}^{P-1} 2^{2\cdot2^j} = 2^{2(2^{P}-1)} \le f^\varepsilon,
\]
which is equivalent to
\[
2(2^{P}-1) = 4\cdot2^{P-1} - 2 = k_1 - 2 \le \varepsilon \log f.
\]
Hence, it is feasible to choose $\varepsilon = k_1/\log f$. Observe that $\sum_{p=1}^{P} k_p = 4(2^P - 1) < 2k_1$. It follows that
\[
 \log f = \log \prod_{p=1}^P (k_p / 2)^{R_p} = 2 \sum_{p=1}^{P} k_p \log (k_p/2) = \Theta(k_1 \log k_1),
\]
implying that $k_1=\Theta(\log f/\log \log f)$. We conclude that $\varepsilon=\Theta(1/\log \log f)$.

Let us now analyse the stabilisation time and space complexity of the construction.
\begin{lemma}
 The algorithm stabilises in $O(f)$ rounds.
\end{lemma}
\begin{proof}
To analyse the stabilisation time, we first bound the stabilisation time of each phase $p$ separately. By \theoremref{thm:layering}, iteration $i+1$ of phase $p$ has stabilisation time
\[
 T_{i+1,p} = T_{i,p} + O\!\left(f_{i+1,p} \cdot k_p^{k_p}\right).
\]
Analogously to the proof of \theoremref{thm:fixed-blocks}, we again get a geometric series and the total stabilisation time of phase $p$ is bounded by
\[
 \sum_{i=1}^{R_p} T_{i+1,p} = O(T_{R_p,p}) = O\!\left(f_p \cdot k_p^{k_p}\right).
\]
Therefore, the total stabilisation time is 
\[
 T = \sum_{p=1}^{P} T_p = O\!\left(\sum_{p=1}^{P} f_p \cdot k_p^{k_p}\right).
\]
For $1 \le p \le P-4$, using the shorthand $\kappa = k_{p+1}$, we can bound
\[
\frac{f_{p+1}}{f_{p}} \cdot \frac{\kappa^{\kappa}}{k_p^{k_p}} = d_{p+1} \frac{\kappa^{\kappa}}{\left(2\kappa\right)^{2\kappa}} 
= \frac{\kappa^{2\kappa}}{2^{2\kappa}}\cdot \frac{\kappa^{\kappa}}{\left(2\kappa\right)^{2\kappa}}
= \left(\frac{\kappa}{16}\right)^{\kappa} \ge 2.
\]
Thus, we get a geometric series
\begin{align*}
\sum_{p=1}^{P-3} T_p \le 2 T_{P-2} = O(f), 
\end{align*}
and since also $T_{P-2} + T_{P-1} + T_{P} = O(f)$, the stabilisation time of all phases is bounded by $O(f)$.
\end{proof}

\begin{lemma}
 Every node uses at most $O(\log^2 f /\log \log f + \log c)$ bits of memory.
\end{lemma}
\begin{proof}
By Theorem~\ref{thm:layering}, the number of state bits increases each iteration by $\Theta(\log C_i)$ bits, where $C_i$ is the counter size needed for iteration $i$. During phase $p$, the counter size is $O(f \cdot k_p^{k_p})$ in each iteration. There are exactly $R_p = 2k_p$ iterations phase $p$, and thus the number of bits we need is bounded by
\[
 S_p = O\!\left( k_p \left(\log f + k_p \log k_p\right) \right) = O\!\left(k_p \log f + k_p^2 \log k_p\right).
\]
From earlier computations, we know that $\sum_{i=1}^P k_p = O(k_1)$ and $k_1 = O(\varepsilon \log f)$. Thus, we use
\begin{align*}
 S = \sum_{p=1}^P S_p &= O\!\left(\sum_{p=1}^P (k_p \log f + k_p^2 \log k_p) \right) \\
   &= O\!\left(\log f \sum_{p=1}^P k_p + \sum_{p=1}^P k_p^2 \log k_p \right) \\
   &= O\!\left(k_1 \log f  + k_1^2 \log k_1 \right) \\
   &= O(\varepsilon \log^2 f + \varepsilon^2 \log^2 f (\log \varepsilon + \log \log f)) \\
   &= O\!\left(\frac{\log^2 f}{\log \log f}\right)
\end{align*}
bits in total, as $\varepsilon = \Theta(1/\log \log f)$. Storing the output of the resulting $c$-counter introduces additional $O(\log c)$ bits.
\end{proof}

\section{Saving on Communication Using Randomization}\label{sec:random}

So far we have considered the model where each node broadcasts its entire state every round. In the case of the algorithm given in \theoremref{thm:adaptive}, every node will send $S = O(\log^2 f / \log \log f + \log c)$ bits in each round. As there are $\Theta(n^2)$ communication links, the total number of communicated bits in each round is $\Theta(n^2 S)$. In this section, we consider a randomised variant of the algorithm that achieves better message and bit complexities in a slightly different communication model.

%###
\subsection{The Pulling Model}
%###
Throughout this section we consider the following model, where in every synchronous round:
\begin{enumerate}[noitemsep]
  \item each processor contacts a subset of other nodes by \emph{pulling} their state,
  \item each contacted node \emph{responds} by sending their state to the pulling nodes,
  \item all processors update their local state according to the received messages.
\end{enumerate}
As before, faulty nodes may respond with arbitrary states that can be different for different pulling nodes. We define the (per-node) message and bit complexities of the algorithm as the maximum number of messages and bits, respectively, pulled by a non-faulty node in any round.

The motivation for this model is that it permits to attribute the energy cost for a message to the pulling node. In a circuit, this means that the pulling node provides the energy for the signal transitions of the communication link: logically, the link is part of the pulling node's circuitry, whereas the ``sender'' merely spends the energy for writing its state into the register from which all its outgoing links read.

Our goal will be to keep the number of pulls by non-faulty nodes small at all times. This way a small energy budget per round per node suffices in correct operation. By limiting the energy supply of each node, we can also effectively limit the energy consumption of the Byzantine nodes.

\subsection{The High-Level Idea of the Probabilistic Construction}

To keep the number of pulls, and thus number of messages sent, small, we modify the construction of Theorem~\ref{thm:layering} to use random sampling where useful. Essentially, the idea is to show that \emph{with high probability} a small set of sampled messages accurately represents the current state of the system and the randomised algorithm will behave as the deterministic one. There are two steps where the nodes rely on information broadcast by the all the nodes: the majority voting scheme over the blocks and our variant of the phase king algorithm. Both can be shown to work with high probability by using concentration bound arguments.

More specifically, for any constant $\kappa \ge 1$ we can bound the probability of failure by $\totaln^{-\kappa}$ by sampling $M = \Theta(\log \eta)$ messages; here $\eta$ denotes the total number of nodes in the system. The idea is to use a union bound over all levels of recursion, nodes, and considered rounds, to show that the sampling succeeds with high probability in all cases. For the randomised variant of Theorem~\ref{thm:layering}, we will require the following additional constraint: when constructing a counter on $N=kn$ nodes, the total number of failures is bounded by $F < \frac{N}{3+\gamma}$, where $\gamma>0$ is some constant. Since the resilience of the recursive construction is suboptimal anyway, this constraint is always going to be satisfied. This allows us to construct \emph{probabilistic synchronous $c$-counters} in the sense that the counter stabilises in time $T$ if for all rounds $t \ge T$ all non-faulty nodes count correctly with probability $1-\eta^{-\kappa}$.

%###
\subsection{Sampling Communication Channels}\label{ssec:sampling}
%###

There are two steps in the construction of \theoremref{thm:layering} where we rely on deterministic broadcasting: the majority sampling for electing a leader block and the execution of the phase king protocol. We start with the latter.

%###
\paragraph{Randomised Phase King.} 
%###
Instead of checking whether at least $N-F$ of all messages have the same value, we check whether at least a fraction of $2/3$ of the sampled messages have the same value. Similarly, when checking for at least $F+1$ values, we check whether a fraction $1/3$ of the sampled messages have this value.

\begin{lemma}\label{lemma:whp}
 Let $x \in [C] \cup \{\infty\}$ and suppose a node samples $M$ values from the other nodes. Then there exists $M_0(\eta,\kappa,\gamma) = \Theta(\log \eta)$ so that $M\geq M_0$ implies the following with high probability.
 \begin{enumerate}[label={(\alph*)}]
  \item If all non-faulty nodes agree on value $x$, then $x$ is seen at least $2/3\cdot M$ times.
  \item If the majority of non-faulty nodes have value $x$, then more than $1/3\cdot M$ sampled values will be $x$.
  \item If at least $2/3\cdot M$ sampled values have value $x$, then $x$ is a majority value.
 \end{enumerate}
\end{lemma}
\begin{proof}
Define $\delta = 1-\frac{2}{3}\cdot\frac{3+\gamma}{2+\gamma}$ and let the random variable $X$ denote the number of $x$ values sampled from \emph{non-faulty} nodes.

(a) If all non-faulty nodes agree on value $x$, then 
\[
\E[X] = \left(1-\frac{F}{N}\right)M > \frac{2+\gamma}{3+\gamma}M.
\]
As $\delta$ satisfies $(1-\delta)\E[X] > 2/3\cdot M$, it follows from Chernoff's bound that
\begin{align*}
 \Pr\left[X < \frac{2}{3}M\right] &\le \Pr[X < (1-\delta) \E[X]] \\
               &\le \exp\left( -\frac{\delta^2}{2} \E[X] \right) \\
               &\le \exp\left( -\delta^2 \frac{2+\gamma}{2(3+\gamma)} M_0 \right).
\end{align*}
For sufficiently large $M_0(N,\kappa,\gamma) = \Theta(\log N)$ this probability is bounded by $N^{-\kappa}$.

 (b) If a majority of non-faulty nodes have value $x$, then $\E[X] \ge \frac{1}{2}\cdot\frac{2+\gamma}{3+\gamma}M$. As above, by picking the right constants and using concentration bounds, we get that
\begin{align*}
 \Pr\left[X \le \frac{1}{3}M\right] &\le \Pr[X < (1-\delta) \E[X]] \\
                &\le \exp\left( -\frac{\delta^2}{2} \E[X] \right) \\
               &\le \exp\left( -\delta^2 \frac{2+\gamma}{4(3+\gamma)} M_0 \right) \le N^{-\kappa}.
\end{align*}

 (c) Suppose the majority of non-faulty nodes have values different from $x$. Defining $\bar{X}$ as the random variable counting the number of samples with values different from $x$ and arguing as for (b), we see that
 \begin{align*}
 \Pr\left[X \geq \frac{2}{3}M\right] &=  \Pr\left[\bar{X} < \frac{1}{3}M\right] \le N^{-\kappa},
 \end{align*}
 where again we assume that $M_0(N,\kappa,\gamma) = \Theta(\log N)$ is sufficiently large. Thus, $X\geq 2/3\cdot M$ implies with high probability that the majority of non-faulty nodes have value~$x$.
\end{proof}

As a corollary, we get that when using the sampling scheme, the execution of the phase king essentially behaves as in the deterministic broadcast case.

\begin{corollary}\label{coro:randomised_phase_king}
When executing the randomised variant of the phase king protocol from \sectionref{sec:construction} for $\eta^{O(1)}$ rounds, the statements of \lemmaref{lemma:establish_agreement} and \lemmaref{lemma:maintain_agreement} hold with high probability.
\end{corollary}
\begin{proof}
The algorithm uses two thresholds, $N-F$ and $F+1$. As discussed, these are replaced by $2/3\cdot M$ and $1/3\cdot M$ when taking $M$ samples. Using the statements of \lemmaref{lemma:whp}, we can argue analogously to the proofs of \lemmaref{lemma:establish_agreement} and \lemmaref{lemma:maintain_agreement}; we apply the union bound over all rounds and samples taken by non-faulty nodes ($N-F<\eta$ per round), i.e., over $\eta^{O(1)}$ events.
\end{proof}

%###
\paragraph{Randomised Majority Voting.}
%###

It remains to handle the case of majority voting in the construction of \theoremref{thm:layering}. Consider some level of the recursive construction, in which we want to construct a counter of $N=kn$ nodes out of $k$ $n$-node counters. If $N \ll \log \eta /\log \log \eta$, we can perform the step in the recursive construction using the deterministic algorithm, that is, pulling from all $kn$ nodes. Otherwise, similar to the above sampling scheme for randomised phase king, each node will from each block uniformly sample $M\geq M_0(\eta,\kappa,\gamma) = \Theta(\log \eta)$ states. Again by applying concentration bounds, we can show that with high probability, the non-faulty nodes sample a majority of non-faulty nodes from non-faulty blocks. Thus, we can get a probabilistic version of \lemmaref{lemma:round-counter}.

Recall from \sectionref{sec:construction} that
\begin{align*}
 b^{i} &= \majority\{ b[i,j] \mid j \in [n] \}, \\
 B    &= \majority\{ b^{i} \mid i \in [k] \}, \\
 R    &= \majority\{ r[B,j] \mid j \in [n] \},
\end{align*}
where the $\majority$ function may output an arbitrary value if there is no majority of non-faulty nodes supporting the same value. Analogously to \sectionref{sec:construction}, we define the following local variables at node $(i,j)$ in round $q$:
\begin{align*}
 b^{i'}_q[i,j] &= \majority\{ b_q[i',j'] \mid (i',j') \mbox{ sampled by }(i,j) \mbox{ in round }q \}, \\
 B_q[i,j]    &= \majority\{ b^{i'}_q[i,j] \mid i' \in [k] \}, \\
 R_q[i,j]    &= \majority\{ r[B_q[i,j],j'] \mid (i',j') \mbox{ sampled by }(i,j) \mbox{ in round }q \}.
\end{align*}
Here we sample with repetition and the above sets are multisets; this means all samples from a block are independent and we can readily apply Chernoff's bound.

\begin{lemma}\label{lemma:leader_probabilistic}
Suppose $x\in [k]$ is a non-faulty block, $M_0 = \Theta(\log \eta)$ is sufficiently large, and all non-faulty blocks count correctly in round $q$. If for all non-faulty blocks $i$ and non-faulty nodes $(i,j)$ it holds that $b[i,j]=x$, then with high probability
\begin{enumerate}
  \item $b^{i'}_q[i,j]=x$ for all non-faulty blocks $i'$,
  \item $B_q[i,j]=x$, and
  \item $R_q[i,j]=r_q[x,j']$ for an arbitrary non-faulty node $(B,j')$ in block $x$.
\end{enumerate}
\end{lemma}
\begin{proof}
Consider a non-faulty block (recall that a block is non-faulty if it has at most $f$ faulty nodes). Let $X$ denote the number of states of non-faulty nodes sampled from this block by $(i,j)$ in round $q$. As $f<n/3$, we have that $\E[X] \ge M(n-f)/n> 2/3\cdot M$. Applying Chernoff's bound for $\delta=1/4$ and choosing sufficiently large $M_0(\eta, \kappa) = \Theta(\log \eta)$, we obtain that
\[
 \Pr[X \le M/2] \le \Pr[X \le (1-\delta) \E[X]] \le \exp\!\left(-\frac{\delta^2}{2} \E[X]\right) \leq \eta^{-\kappa}.
\]
Applying the union bound to all nodes and all blocks, it follows that, with high probability, non-faulty nodes always sample a majority of non-faulty nodes from non-faulty blocks. The first statement follows, immediately yielding the second as a majority of the blocks is non-faulty. The third statement now holds because we assume that non-faulty blocks count correctly and $x$ is non-faulty.
\end{proof}

%###
\subsection{Randomised Resilience Boosting}\label{ssec:randomised-layering}
%###

Define $\mathcal{P}(n, f, c, \eta, \kappa)$ as the family of probabilistic synchronous $c$-counters on $n$ nodes and resilience $f$, where probabilistic means that an algorithm $\vec P\in \mathcal{P}(n, f, c, \eta, \kappa)$ of stabilisation time $T(\vec P)$ merely guarantees that it counts correctly with probability $1-\eta^{-\kappa}$ in rounds $t \ge T(\vec P)$. This means that with high probability, eventually all non-faulty nodes agree on a common clock for sufficiently many rounds. Together with \corollaryref{coro:randomised_phase_king}, we obtain a randomized variant of \theoremref{thm:layering}.

\pagebreak[2]

\begin{theorem}\label{thm:layering_randomised}
Given $n,f,\eta \in \N$, pick new parameters $N,F,C\in \N$ and $\kappa>0$, where
\begin{itemize}[noitemsep]
  \item the number of nodes $N=kn\leq \eta$ for some number of blocks $3\leq k\in \N$,
  \item the resilience $F< (f+1)m$, where we abbreviate $m=\lceil k/2\rceil$,
  \item $C>1$ is the new counter size, and
  \item $\kappa$ is a constant.
\end{itemize}
Choose any $c\in \N$ that is an integer multiple of $3(F+2) (2m)^k$. Then for any $\vec{A}\in \mathcal{A}(n, f, c)$, there exists $\vec P \in \mathcal{P}(N,F,C,\eta,\kappa)$ with the following properties.
\begin{enumerate}[noitemsep]
 \item $T(\vec P) = T(\vec A) + 3(F+2) (2m)^k$, and
 \item $S(\vec P) = S(\vec A) + \left\lceil \log (C+1) \right\rceil + 1$.
 \item Each node pulls $O(k\log \eta)$ messages in each round.
\end{enumerate}
\end{theorem}
Note that we can choose to replace $\vec{A}\in \mathcal{A}(n, f, c)$ by $\vec{Q}\in \mathcal{P}(n, f, c, \eta, \kappa)$ when applying this theorem, arguing that with high probability it \emph{behaves} like a corresponding algorithm $\vec{A}\in \mathcal{A}(n, f, c)$ for polynomially many rounds. Applying the recursive construction from \sectionref{sec:recursion} and the union bound, this yields \corollaryref{cor:random-alg}. By always choosing $k = O(\log \eta)$, each node pulls $O(\log^2 \eta)$ messages from other nodes for each layer.

\begin{corollary}\label{cor:random-alg}
For any $c>1$, there exist probabilistic synchronous $c$-counters with a resilience of $f=n^{1-o(1)}$ that stabilise in $O(f)$ rounds, use $O(\log^2 f  / \log \log f + \log c)$ bits of space per node, and in which each node pulls $O(\log \eta (\log f / \log \log f)^2)$ messages per round.
\end{corollary}

We note that it is also possibility to boost the probability of success, and thus the period of stability, by simply increasing the sample size. For instance, sampling $\polylog \eta$ messages yields an error probability of $\eta^{-\polylog \eta}$ in each round, whereas in the extreme case, by ``sampling'' all nodes the algorithm reduces to the deterministic case.

\subsection{Oblivious Adversary}

Finally, we remark that under an \emph{oblivious adversary}, that is, an adversary that picks the set of faulty nodes independently of the randomness used by the non-faulty nodes, we get \emph{pseudo-random} synchronous counters satisfying the following: (1) the execution stabilises with high probability and (2) if the execution stabilises, then all non-faulty nodes will deterministically count correctly. Put otherwise, we can fix the random bits used by the nodes to sample the communication links \emph{once}, and with high probability we sample sufficiently many communication links to non-faulty nodes for the algorithm to (deterministically) stabilise. This gives us the following result.

\begin{corollary}
For any $c>1$, there exist pseudo-random synchronous $c$-counters with a resilience of $f=n^{1-o(1)}$ against an oblivious fault pattern that stabilise in $O(f)$ rounds with high probability, use $O(\log^2 f  / \log \log f + \log c)$ bits of space per node, and in which each node pulls $O(\log \eta (\log f / \log \log f)^2)$ messages per round.
\end{corollary}

\section{Conclusions}

In this work, we showed that there exist (1) \emph{deterministic} algorithms for synchronous counting that have (2) linear stabilisation time, (3) use a very small number of state bits while still achieving (4) almost-optimal resilience--something no prior algorithms have been able to do. In addition, we discussed how to reduce the total number of communicated bits in the network, while still achieving (2)--(4) by considering probabilistic and pseudo-random synchronous counters. 

We conclude by highlighting a few open problems:
\begin{enumerate}
  \item Are there randomised or deterministic algorithms with the optimal resilience of $f < n/3$ that use $\polylog f$ state bits and stabilise in $O(f)$ rounds?
  \item Are there deterministic algorithms that use substantially fewer than $\log^2 f$ state bits?
  \item Are there communication-efficient and space-efficient algorithms with high resilience that stabilise quickly in the usual synchronous model?
\end{enumerate}

% --- Back matter ----

\DeclareUrlCommand{\Doi}{\urlstyle{same}}
\renewcommand{\doi}[1]{\href{http://dx.doi.org/#1}{\footnotesize\sf doi:\Doi{#1}}}
\bibliographystyle{plainnat}
\bibliography{counting}

\end{document}